\documentclass[11pt,b5paper,twoside]{article}      %fleqn reqno ,draft

\newcommand{\mechyear}{2010}
\newcommand{\mechvol}{40}

\usepackage[cp1251]{inputenc}
\usepackage[T2A]{fontenc}
\usepackage[russian]{babel}

\usepackage[dvips]{graphicx}
\usepackage[dvips]{color}

\usepackage{amsmath}
\usepackage{amssymb}
\usepackage{amsxtra}
\usepackage{latexsym}
\usepackage{ifthen}
\usepackage[centerlast,footnotesize,sl]{caption2}
\usepackage{bm}
\usepackage{exscale}

\usepackage{mtt37}

\usepackage{array}
\usepackage{multirow}
\usepackage{hhline}

\newcommand{\mo}{m_0}
\newcommand{\no}{n_0}
\newcommand{\vo}{v_0}
\newcommand{\ri}{{\rm i}\,}

\newcommand {\rk} {\mathop{\rm rank}\nolimits}
\newcommand {\lsgn} {\mathop{\rm bsgn}\nolimits}
\newcommand{\FF}{\mathcal{F}_\mathbf{f}}
\newcommand {\cf}{\mathbf{f}}
\newcommand {\ts}[1] {\textsl{#1}}
\newcommand {\mP}{\mathcal{P}}
\newcommand {\bT}{\mathbf{T}}
\newcommand {\bV}{\mathbf{V}}
\newcommand {\bW}{\mathbf{W}}
\newcommand {\bZ}{\mathbf{Z}}
\newcommand {\bX}{\mathbf{X}}

\newcommand {\bU}{\mathbf{Y}}
\newcommand {\bu}{\mathbf{y}}
\newcommand {\vu}{y}

\newcommand {\bv}{\mathbf{v}}
\newcommand {\bw}{\mathbf{w}}
\newcommand {\bz}{\mathbf{z}}
\newcommand {\bbT}{\mathbb{T}}
\newcommand {\bbR}{\mathbb{R}}
\newcommand {\bbZ}{\mathbb{Z}_2}
\newcommand {\ds}{\displaystyle}
\newcommand {\mstrut}{\vphantom{\bigl(}}
\newcommand {\rT}{\mathrm{T}}
\newcommand{\ru}{\rule{0pt}{12pt}}
\newcommand{\vk}{\varkappa}
\newcommand{\gs}{\geqslant}
\newcommand{\ls}{\leqslant}
\newcommand{\bo}{{\boldsymbol \omega}}
\newcommand{\ba}{{\boldsymbol \alpha}}
\newcommand{\bb}{{\boldsymbol \beta}}

\begin{document}

\udk{531.38}

\title{ОБОБЩЕНИЕ 4-ГО КЛАССА АППЕЛЬРОТА: \\ФАЗОВАЯ ТОПОЛОГИЯ}
      {Обобщение 4-го класса Аппельрота: фазовая топология}

\author{М.П.~Харламов }

\thanks{Работа выполнена при финансовой поддержке гранта РФФИ № 10-01-00043.}

\date{01.07.10}

\address{Волгоградская академия гос. службы, Россия}

\email{mharlamov@vags.ru}

\maketitle

\begin{abstract}
Статья продолжает цикл работ автора (Механика твердого тела,
вып.~35, 2005 и вып.~38, 2008), в которых исследуется интегрируемая
динамическая система на четырехмерном инвариантном подмножестве
фазового пространства задачи о движении твердого тела в двойном
силовом поле. При стремлении к нулю напряженности одного из полей
эта система обращается в семейство особо замечательных движений
\mbox{4-го} класса Аппельрота волчка Ковалевской в поле силы
тяжести. Предложен метод описания фазовой топологии при наличии
алгебраических зависимостей фазовых переменных от переменных
разделения с использованием булевых вектор-функций. Выполнен грубый
топологический анализ рассматриваемой системы с двумя степенями
свободы.
\end{abstract}

%\tableofcontents

\Section{Исходные соотношения}\label{sec1} В
интегрируемой системе с тремя степенями свободы, описывающей
движение волчка типа Ковалевской в двойном силовом поле
\cite{ReySem}, в работе \cite{Kh34} найдены все так называемые
критические подсистемы -- инвариантные многообразия в фазовом
пространстве $\mP^6 \subset \bbR ^9(\ba,\bb,\bo)$, почти всюду
четырехмерные, на которых индуцируются гамильтоновы системы с двумя
степенями свободы и объединение которых (с точностью до
однопараметрического семейства маятниковых движений вокруг оси
динамической симметрии, формирующего критическую подсистему с одной
степенью свободы) составляет критическое множество трех первых
интегралов на $\mP^6$. Далее рассматривается критическая подсистема
на подмногообразии $\mathfrak{O} \subset \mP^6$, изучавшаяся в
работах \cite{Kh35,Kh38}. В \cite{Kh35} указаны частные интегралы
$S$ и $\rT$, образующие полную инволютивную систему на
$\mathfrak{O}$, построена их бифуркационная диаграмма, изучены
проекции интегральных многообразий на плоскость вспомогательных
переменных. В результате получено {\it допустимое множество} --
образ фазового пространства $\mathfrak{O}$ рассматриваемой
подсистемы под действием интегрального отображения $\mathcal{F}=S{\times}\rT: \mathfrak{O} \to \bbR ^2$.
На основе геометрии
проекций предложена система координат (вообще говоря, комплексная),
в которых переменные должны разделяться. Эти результаты получили
дальнейшее развитие в работе \cite{Kh38}, в которой выполнено
фактическое разделение переменных в {\it вещественной} области и
построены алгебраические выражения фазовых переменных через
переменные разделения. Приведем формулы, необходимые в дальнейшем.

Пусть $a,b$ -- постоянные модули векторов $\ba,\bb$ ($a>b>0$),
$s,\tau$ -- константы интегралов $S,\rT$. Введем вспомогательные
константы $p,r,\sigma,\chi,\vk$, полагая ${p>0},{r>0},{\chi \gs 0}$ и
\begin{equation}\label{eq1_1}
\begin{array}{c}
p^2=a^2+b^2, \quad r^2=a^2-b^2, \\
\sigma=\tau^2-2p^2\tau+r^4, \quad 4s^2 \chi^2 = \sigma+4s^2\tau,
\quad \vk =\sqrt{\sigma}.
\end{array}
\end{equation}
Фиксируем интегральное многообразие
$$
\FF =\{S=s,\rT=\tau\} \subset \mathfrak{O}, \qquad \cf=(s,\tau)
$$
и вводим переменные разделения $t_1,t_2$, полагая
\begin{equation}\notag
t_1= \ds{\frac{\tau y+ xz}{\tau-x^2}}, \qquad t_2= \ds{\frac{\tau y-
xz}{\tau-x^2}},
\end{equation}
где
$$
\begin{array}{l}
x^2=(\alpha_1-\beta_2)^2+(\alpha_2+\beta_1)^2, \\
y^2=p^2-2(\alpha_1 \beta_2-\alpha_2 \beta_1)-\tau, \\
z^2=\sigma
x^2+\tau y^2-\sigma \tau.
\end{array}
$$
Тогда, как показано в \cite{Kh38}, для конфигурационных переменных
будем иметь
\begin{equation}
\begin{array}{l}
\displaystyle{\alpha_1=(U_1-U_2)^2\frac{(\mathcal{A}-r^2 U_1 U_2)(4
s^2 \tau+U_1 U_2)-(\tau+r^2) s\tau P_1 V_1 P_2 V_2}{4 r^2
s\, \tau (t_1^2-t_2^2)^2},
} \\[3mm]
\displaystyle{\alpha_2=\ri (U_1-U_2)^2 \frac{(\mathcal{A} -r^2 U_1
U_2)s\tau V_1 V_2 -(4 s^2 \tau+U_1 U_2)(\tau+r^2) P_1 P_2}{4
r^2 s\, \tau (t_1^2-t_2^2)^2},
}\\[3mm]
\displaystyle{\alpha_3= \frac{R }{2 r } \, \frac {M_1
M_2}{t_1+t_2},} \\
\displaystyle{\beta_1=\ri (U_1-U_2)^2
\frac{(\mathcal{B}+r^2 U_1 U_2)s\tau  V_1 V_2-(4 s^2 \tau+U_1
U_2)(\tau-r^2) P_1 P_2}{4 r^2 s\, \tau (t_1^2-t_2^2)^2},
} \\[3mm]
\displaystyle{\beta_2=-(U_1-U_2)^2 \frac{(\mathcal{B} + r^2 U_1
U_2)(4 s^2 \tau + U_1 U_2)-(\tau-r^2) s\tau  P_1 V_1 P_2
V_2}{4 r^2 s\, \tau (t_1^2-t_2^2)^2},
}\\[3mm]
\displaystyle{\beta_3= - \ri \frac{R}{2 r} \, \frac { N_1
N_2}{t_1+t_2},}
\end{array}\label{eq1_3}
\end{equation}
а для угловых скоростей получим выражения
\begin{equation}
\begin{array}{l} \displaystyle{\omega_1=  \ri \frac{R }{4 r s\,
\sqrt{2}} \, \frac{M_2 N_1
U_1 V_2 + M_1 N_2 U_2 V_1}{ t_1^2-t_2^2},} \\[3mm]
\displaystyle{\omega_2= \frac{R}{4 r s\, \sqrt{2} } \,
\frac{M_2 N_1 U_2 V_1 + M_1 N_2 U_1 V_2}{t_1^2-t_2^2},} \\[3mm]
\displaystyle{\omega_3=  \ri \frac{{U_1-U_2}}{\sqrt{2}}\frac{M_2 N_2
V_1 - M_1 N_1 V_2}{t_1^2-t_2^2}.}
\end{array}\label{eq1_4}
\end{equation}
Здесь введены алгебраические радикалы
\begin{equation}\label{eq1_2}
\begin{array}{ll}
K_1=\sqrt{\mstrut t_1+\varkappa}\,, & K_2=\sqrt{\mstrut t_2+\varkappa}\,,\\
L_1=\sqrt{\mstrut t_1-\varkappa}\,, & L_2=\sqrt{\mstrut t_2-\varkappa}\,,\\
M_1=\sqrt{\mstrut t_1+\tau+r^2}\,, & M_2=\sqrt{\mstrut
t_2+\tau+r^2}\,,
\\[1.5mm]
N_1=\sqrt{\mstrut t_1+\tau-r^2}\,, &
N_2=\sqrt{\mstrut t_2+\tau-r^2}\,,\\
V_1 = \sqrt{\mstrut \ds{\frac{t_1^2-4s^2\chi^2}{s\tau}}}\,, & V_2 =
\sqrt{\mstrut \ds{\frac{t_2^2-4s^2\chi^2}{s\tau}}}
\end{array}
\end{equation}
и для сокращения записи обозначено
\begin{equation}\notag
\begin{array}{l}
U_1=K_1 L_1, \quad U_2=K_2 L_2,\quad R=K_1 K_2 + L_1 L_2,\\
P_1=M_1 N_1, \quad P_2=M_2 N_2,\\
\mathcal{A}=[(t_1+\tau+r^2)(t_2+\tau+r^2)-2(p^2+r^2)r^2]\tau, \\
\mathcal{B}=[(t_1+\tau-r^2)(t_2+\tau-r^2)+2(p^2-r^2)r^2]\tau.
\end{array}
\end{equation}
При этом переменные $t_1,t_2$ удовлетворяют дифференциальным
уравнениям
\begin{equation}\label{eq1_5}
\begin{array}{ll}
\displaystyle{(t_1-t_2)\frac{d t_1}{d t} = \sqrt{\frac{1}{2 s
\tau}(t_1^2-4s^2\chi^2)(t_1^2-\sigma)[(t_1+\tau)^2-r^4]}\, ,} \\
[3mm] \displaystyle{(t_1-t_2)\frac{d t_2}{d t} = \sqrt{\frac{1}{2 s
\tau}(t_2^2-4s^2\chi^2)(t_2^2-\sigma)[(t_2+\tau)^2-r^4]}\,.}
\end{array}
\end{equation}

Связную компоненту множества, в котором при заданных $s,\tau$
осциллируют переменные $t_1,t_2$, будем называть достижимой
областью. Достижимая область здесь является либо прямоугольником,
либо произведением отрезка на дополнение к интервалу. Для внутренних
точек достижимой области выражения \eqref{eq1_3}, \eqref{eq1_4}
определят, вообще говоря, $2^{10}$ различных точек фазового
пространства. В то же время, для решения задачи грубого
топологического анализа необходимо для каждой достижимой области
указать количество накрывающих ее связных компонент интегрального
многообразия (торов Лиувилля). В критических случаях, когда точка
$\cf=(s,\tau)$ принадлежит бифур\-кационной диаграмме, в достижимых
областях возникают особенности. Знание количества компонент
связности критической интегральной поверхности и близлежащих
регулярных интегральных многообразий даст и полную информацию о
характере имеющихся бифуркаций. Непосредственное исследование
уравнений \eqref{eq1_3}, \eqref{eq1_4} является крайне громоздким
и вряд ли разумно. Ниже предложен простой способ решения
поставленной задачи, основанный на некоторых стандартных действиях с
двоичными матрицами.

\Section{Метод булевых функций}\label{sec2} Найденные зависимости
фазовых переменных от переменных разделения представляют собой
полиномы от радикалов \eqref{eq1_2}, коэффициенты которых --
однозначные функции от $\bT=(t_1,t_2)$. Обозначим через $\bX$
фазовый вектор, через $\bU$ -- вектор, составленный из радикалов \eqref{eq1_2}.
Компоненты $Y_i$ входят в выражения для $\bX$ в виде
произведений-мономов (причем только в степени 1, так как $Y_i^2$ --
однозначная функция). Представим все такие мономы в виде одного
вектора $\bZ$. Имеем
\begin{equation}\label{eq2_1}
\bX=\mathbf{A} + B \, \bZ,
\end{equation}
где вектор $\mathbf{A}$ и матрица $B$ -- однозначные функции от
$\bT$.

\begin{remark}[1] {\it Далее мы говорим о знаках радикалов \eqref{eq1_2},
фигурирующих как сомножители в мономах $Z_i$. В то же время эти
радикалы могут быть как вещественными, так и чисто мнимыми. Однако,
выбрав для рассмотрения некоторую достижимую область, меняя при
необходимости знаки подкоренных выражений в радикалах и вынося
возникающую мнимую единицу в коэффициенты, можно добиться того,
чтобы все величины в $\eqref{eq2_1}$ были вещественными. Будем
иметь это в виду при рассмотрении конкретных областей. На общий ход
рассуждений и результат подобные преобразования в различных областях не влияют.}
\end{remark}

Малым шевелением $\cf$ сделаем все торы в составе $\FF$
нерезонансными, и выберем $\bT$ внутри достижимой области так, что
зависимости \eqref{eq2_1} будут знакоопределенными в следующем
смысле: замена знака у любого из мономов $Z_i$ меняет точку $\bX$.
Иначе говоря, при выбранном (и далее фиксированном) $\bT$ количество
точек $\bX$, заданных равенством \eqref{eq2_1}, равно $2^k$, где $k={\dim
\bZ}$. Пусть $\mathbf{X}_0$ -- одна из точек \eqref{eq2_1}, а
$\FF(\bX_0)$ -- связная компонента $\FF$, содержащая $\bX_0$.
Выпустим траекторию из точки $\bX_0$. Вдоль нее часть радикалов
$Y_i$ имеет постоянный знак (назовем их радикалами первой группы), а
другая часть этих же радикалов периодически меняет знак (эту часть называем второй
группой). Как правило, вторая группа имеет по два радикала на каждую
из переменных разделения. Соберем первую группу в вектор
$\bV=(V_1,\ldots,V_m)$, а вторую -- в вектор $\bW=(W_1,\ldots,W_n)$.
Выбранный в начальный момент знак любого $V_i$ сохраняется вдоль
траектории, в то время как любой набор знаков величин $W_j$ будет на
этой траектории получен. Отсюда сразу же следует, что если точка
\eqref{eq2_1} отличается от $\bX_0$ только знаками радикалов второй
группы, то она принадлежит $\FF(\bX_0)$. Более того, если у точек
$\bX_0$ и $\bX$, полученных при заданном $\bT$, различны знаки
радикалов первой группы, но, меняя знаки радикалов второй группы,
можно из набора мономов $\bZ_0$ точки $\bX_0$ получить набор мономов
$\bZ$ точки $\bX$, то $\bX \in \FF(\bX_0)$. Формализуем эти
рассуждения в терминах булевых функций.

Введем функцию $\lsgn: \bbR \to \bbZ$, которую будем называть булевым знаком:
\begin{equation}\notag
x=(-1)^{\lsgn(x)}|x|,\qquad x \in \bbR.
\end{equation}
Очевидна
мультипликативно-аддитивная двойственность:
\begin{equation}\notag
\lsgn (x_1 x_2) = \lsgn (x_1) \oplus \lsgn (x_2)
\end{equation}
(сумма по модулю 2). Для введенных выше радикалов и мономов
обозначим
\begin{equation}\notag
\vu_i=\lsgn(Y_i), \quad v_i=\lsgn(V_i), \quad w_i=\lsgn(W_i), \quad
z_i=\lsgn(Z_i).
\end{equation}
Имеем
$$
Z_i=\prod_{j=1}^{m+n}Y_j^{c_{ij}}, \quad c_{ij}\in \bbZ \qquad \Rightarrow \qquad z_i=\bigoplus_{j=1}^{m+n}{c_{ij}}\vu_j.
$$
Поэтому многозначная зависимость \eqref{eq2_1} описывается
$\bbZ$-линейным отображением $C$:
\begin{equation}\label{eq2_2}
\bz = C \bu, \qquad \bu \in \bbZ^{m+n}, \quad \bz \in \bbZ^k.
\end{equation}
В соответствии с разбиением радикалов на группы аргумент имеет вид
\begin{equation}\label{eq2_3}
\bu =(\bv,\bw), \qquad \bv \in \bbZ^m, \quad \bw \in \bbZ^n.
\end{equation}
При этом, фиксируя $\bv$, порождающее точку $\bX_0$, и меняя $\bw$
произвольно, получим двоичные наборы $\bz$, которым отвечают точки
одного тора в составе $\FF$. Поскольку при заданном
$\bT$ точки фазового пространства однозначно зависят от $\bz$, то
множество точек над $\bT$, принадлежащих связной компоненте
$\FF(\bX_0)$, находится во взаимно-однозначном соответствии с
множеством $\{C(\bv,\bw):\bw \in \bbZ^n\}$. Эти множества для
различных $\bv$ либо не пересекаются, либо совпадают, порождая тем
самым отношение эквивалентности в $\bbZ^m$. Количество $c(\cf)$
классов эквивалентности равно количеству связных компонент в $\FF$,
накрывающих выбранную достижимую область (связную область осцилляции
переменных разделения).
\begin{lemma}[1] Пусть $Y,Z$ -- векторные пространства $($над любым полем
$F)$ и $C: Y \to Z$ -- линейное отображение. Пусть $Y=V{\times}W$
-- прямое произведение векторных пространств $V,W$. Определим
отношение эквивалентности в $V$, полагая
\begin{equation}\notag
v'\sim v'' \Leftrightarrow \exists \, w',w'' \in W: \;
C(v',w')=C(v'',w''),
\end{equation}
и пусть $K_v$ -- класс эквивалентности элемента $v\in V$. Тогда

$1^\circ$ $K_0$ -- подпространство в $V$;

$2^\circ$ $K_v=v+K_0$;

$3^\circ$ множество классов эквивалентности как векторное
пространство над $F$ имеет размерность $d=P-Q$, где $P$ -- ранг
ограничения отображения $C$ на подпространство $V{\times}\{0\}$, а
$Q=\dim R$, где $R=C(V{\times}\{0\}) \cap C(\{0\}{\times}W)$.
\end{lemma}
\begin{proof} Утверждения $1^\circ$ и $2^\circ$ следуют из
линейности $C$. Тогда множество классов эквивалентности
отождествляется с фактор-про\-стран\-ством $V/K_0$. Поскольку $v \sim 0$
равносильно существованию некоторого $w\in W$, такого, что
$C(v,0)=C(0,w)$, то $V/K_0 \cong {C(V{\times}\{0\})/R}$, откуда
и следует утверждение $3^\circ$.
\end{proof}

Возвращаясь к отображению \eqref{eq2_2}, запишем его матрицу согласно
разбиению~\eqref{eq2_3} в виде $C=\left(C_{\bv}|
C_{\bw} \right)$, где $C_{\bv}$ и $C_{\bw}$ -- двоичные матрицы
размерности $k{\times}m$ и $k{\times}n$ соответственно. Очевидно,
элементарные преобразования строк $C$ и столбцов $C_{\bv}$ или
$C_{\bw}$, которые можно трактовать как автоморфизмы пространств
$Z=\{\bz\}=\bbZ^k$, $V=\{\bv\}=\bbZ^m$, $W=\{\bw\}=\bbZ^n$ или как
замены базисов в них, не изменяют количества классов
эквивалентности. Над полем $\bbZ$ такие преобразования сводятся к
перестановке строк (столбцов одной группы) или прибавлением к одной
строке другой строки (к одному столбцу другого столбца той же
группы). В результате матрица $C$ может быть приведена к виду
\begin{equation}\label{eq2_4}
C=
\begin{tabular}{||m{1.2cm}|m{2cm}|m{1.3cm}|m{1.2cm}||}

\multirow{2}{1.2cm}{\centering $E_P$} & \multirow{2}{2cm}{\centering
$0_{P,m-P}$} & \multirow{1}{1.3cm}{\centering $0_{P-Q,Q}$} &
\multirow{2}{1.2cm}{\centering $0_{P,n-Q}$} \\
\hhline{||~|~|-|~||}

{} & {} & \multirow{1}{1.3cm}{\centering $E_Q$} & {} \\
\hhline{||-|-|-|-||}

\multirow{1}{1.2cm}{\centering $0_{k-P,P}$} &
\multirow{1}{2cm}{\centering $0_{k-P,m-P}$} &
\multirow{1}{1.3cm}{\centering $0_{k-P,Q}$} &
\multirow{1}{1.2cm}{\centering $\cdots$}\\

\end{tabular}\, ,
\end{equation}
где $P=\rk C|_{\bbZ^m}$, $Q=\dim (C(\bbZ^m)\cap C(\bbZ^n))$, через
$0_{i,j}$ обозначена нулевая $i{\times}j$-матрица, а сомножители в
разложении $\bbZ^{m+n}={\bbZ^m}{\times}{\bbZ^n}$ естественным
образом отождествляются с подпространствами. Последние $k-P$ строк,
в которых на месте многоточия могут содержаться ненулевые элементы,
являются теми координатами преобразованной булевой вектор-функции,
которые зависят только от аргументов второй группы. Изменение этих
аргументов на количество классов эквивалентности, конечно, не
влияет. Поэтому такие строки могут быть отброшены, после чего могут
быть отброшены и нулевые столбцы, и все образовавшиеся нулевые
строки. Останется матрица с единичными блоками $E_P,E_Q$, после чего
количество классов эквивалентности, совпадающее с количеством
компонент связности многообразия $\FF$, накрывающих выбранную
достижимую область, определяется по лемме~1 как ${2^{P-Q}}$. В
частности, это означает, что строки, содержащие единичный блок
$E_Q$, также не влияют на количество классов эквивалентности и могут
быть отброшены. В конечном счете, этими преобразованиями матрица $C$
приводится к единичной $E_{P-Q}$, содержащей только столбцы первой
группы. На практике для определения значения $P-Q$ достаточно
исключить все переменные второй группы и выбрать линейно независимые
строки (например, привести редуцированную матрицу к трапецевидной).

Итак, вычисление количества торов Лиувилля в регулярных интегральных
многообразиях можно провести, применяя к строкам и столбцам матрицы
$C$ метод Гаусса над полем $\bbZ$. Для критических случаев этот же
метод работает с небольшим уточнением процедуры разбиения радикалов
на группы.

\Section{Достижимые области и интегральные многообразия}\label{sec3}
Как следует из уравнений  \eqref{eq1_3}--\eqref{eq1_5}, разделяющим
множеством при классификации траекторий и интегральных многообразий
по параметрам $s,\tau$, в дополнение к очевидной особенности $s\, \tau=0$,
служит дискриминантное множество многочлена
\begin{equation}\label{eq3_1}
S(\xi)=(\xi^2-4s^2\chi^2)(\xi^2-\sigma)[(\xi+\tau)^2-r^4],
\end{equation}
которое состоит из прямых $s=\pm a$, $s=\pm b$, $\tau=(a\pm b)^2$ и
кривых, заданных уравнением $\chi=0$. В явном виде эти кривые можно
записать как ${\tau = \varphi_{\pm}(s)}$, где
$$
\varphi_{\pm}(s)=a^2+b^2-2s^2 \pm \sqrt{\mstrut
(s^2-a^2)(s^2-b^2)}.
$$
Следующее утверждение вытекает из результатов работы \cite{Kh35}.
\begin{theorem}[1]\label{th1}
Допустимое множество на плоскости $(s,\tau)$ состоит из следующих
подмножеств:
\begin{equation}\notag
\begin{array}{l}
\{-a \ls s \ls -b, \, (a-b)^2 \ls \tau < +\infty   \};\\
\{-b < s <0, \,  \varphi_+(s) \ls \tau < +\infty   \};\\
\{0 < s < b, \,  -\infty< \tau \ls \varphi_-(s) \};\\
\{b \ls s \ls a, \, -\infty < \tau \ls (a+b)^2 \};\\
\{s > a, \, \varphi_+(s) \ls \tau \ls (a+b)^2     \}.
\end{array}
\end{equation}
\end{theorem}

\begin{figure}[htp]
\centering
\includegraphics[width=70mm,keepaspectratio]{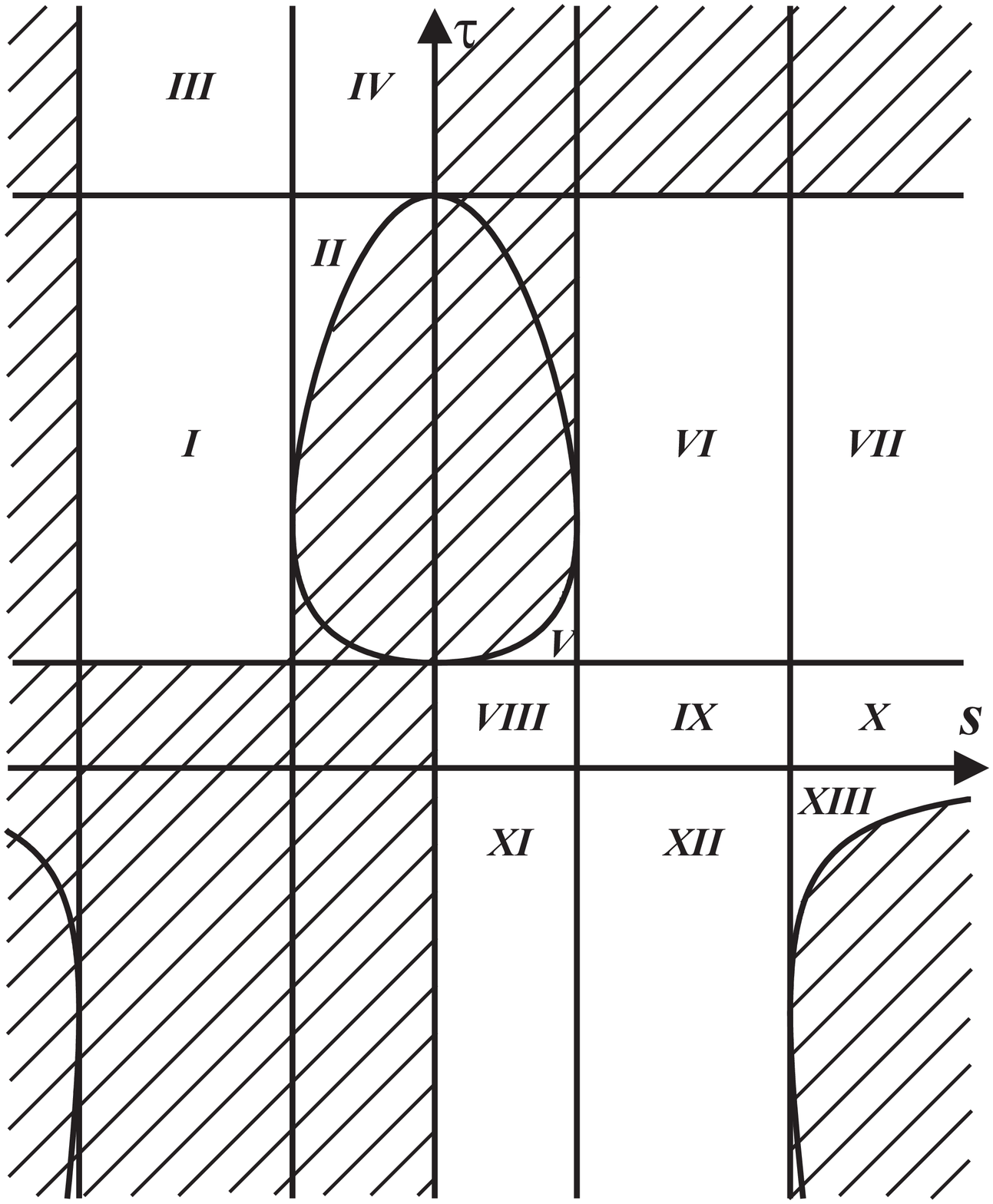}
\caption{Разделяющие кривые и кодировка областей.}\label{fig_harl1}
\end{figure}

\begin{table}[htp]
\centering \small
\begin{tabular}{|c| c| c| c|}
\multicolumn{4}{r}{Таблица 1}\\
\hline
\begin{tabular}{c}Область\end{tabular}&\begin{tabular}{c}Корни $S(\xi)$\end{tabular}&\begin{tabular}{c}Изменение $t_1$\end{tabular}
&\begin{tabular}{c}Изменение $t_2$\end{tabular}\\
\hline \ru  \ts{I} & $\mo<-\vo<\no<\vo$ & $[\,\mo,-\vo\,]$ & $[\,\no,\vo\,]$\\
\hline \ru \ts{II} & $\mo<\no<-\vo<\vo$ & $[\,\mo,\no\,]$ & $[\,-\vo,\vo\,]$\\
\hline \ru \ts{III} & $\mo<-\vo<\no<-\vk<\vk<\vo$ & $[\,\mo,-\vo\,]$ & $[\,\no,-\vk\,]$\\
\hline \ru \ts{IV} & $\mo<\no<-\vo<-\vk<\vk<\vo$ & $[\,\mo,\no\,]$ & $[\,-\vo,-\vk\,]$\\
\hline \ru \ts{V} & $\mo<-\vo<\vo<\no$ & $[\,-\vo,\vo\,]$ & $[\,\no \,(\pm \infty)\, \mo \,]$\\
\hline \ru \ts{VI} & $\mo<-\vo<\no<\vo$ & $[\,-\vo,\no\,]$ & $[\,\vo \,(\pm \infty)\,\mo\,]$\\
\hline \ru \ts{VII} & $-\vo<\mo<\no<\vo$ & $[\,\mo,\no\,]$ & $[\,\vo\,(\pm \infty)\,-\vo\,]$\\
\hline \ru \ts{VIII} & $\mo<-\vo<-\vk<\vk<\vo<\no$ & $[\,\vk,\vo\,]$ & $[\,\no\,(\pm \infty)\,\mo\,]$\\
\hline \ru \ts{IX} & $\mo<-\vo<-\vk<\vk<\no<\vo$ & $[\,\vk,\no\,]$ & $[\,\vo\,(\pm \infty)\,\mo\,]$\\
\hline \ru \ts{X} & $-\vo<\mo<-\vk<\vk<\no<\vo$ & $[\,\vk,\no\,]$ & $[\,\vo\,(\pm \infty)\,-\vo\,]$\\
\hline \ru \ts{XI} & $-\vk<-\vo<\mo<\no<\vo<\vk$ & $[\,\mo,\no\,]$ & $[\,\vo,\vk\,]$\\
\hline \ru \ts{XII} & $-\vk<-\vo<\mo<\vo<\no<\vk$ & $[\,\mo,\vo\,]$ & $[\,\no,\vk\,]$\\
\hline \ru \ts{XIII} & $-\vk<\mo<-\vo<\vo<\no<\vk$ & $[\,-\vo,\vo\,]$ & $[\,\no,\vk\,]$\\
\hline
\end{tabular}
\end{table}

На рис.~\ref{fig_harl1} показаны области \ts{I}--\ts{XIII} (области
регулярности интегрального отображения $\mathcal{F}$), исчерпывающие допустимое
множество и имеющие различное распределение корней многочлена
\eqref{eq3_1}. В дополнение к обозначениям \eqref{eq1_1} положим
$$
\mo=-\tau-r^2, \qquad \no = -\tau+r^2, \qquad  \vo =2|s|\chi.
$$
Требование вещественности выражений \eqref{eq1_3}, \eqref{eq1_4}
определяет промежутки осцилляции переменных $t_1,t_2$. Информация о
расположении корней многочлена $S(\xi)$ и достижимых областях для
переменных разделения сведена в табл.~1. В последнем столбце этой таблицы запись вида $[c_0 (\pm
\infty) c_1]$ означает, что на самом деле $c_1<c_0$, и переменная осциллирует на $\bbR\backslash
(c_1,c_0)$, периодически пересекая бесконечность.

Занумеруем радикалы \eqref{eq1_2} следующим образом
\begin{equation}\notag
R_{i 1}=K_i,\;R_{i 2}=L_i,\;R_{i 3}=M_i,\;R_{i 4}=N_i,\;R_{i 5}=V_i
\qquad  (i=1,2)
\end{equation}
и введем булевы переменные (здесь уместно напомнить замечание~1)
\begin{equation}\notag
\vu_\gamma = \lsgn(R_{1 \gamma}),\qquad  \vu_{5+\gamma} = \lsgn(R_{2 \gamma})
\qquad (\gamma =1, \ldots, 5).
\end{equation}

Выделим в выражениях \eqref{eq1_3} компоненты булевой
вектор-функции \eqref{eq2_2}, отвечающие используемым мономам от
радикалов (в скобках записаны сами мономы)
\begin{equation}\notag
\begin{array}{lll}
z_1  = \vu_1  \oplus \vu_2  \oplus \vu_6  \oplus \vu_7   &{}& (K_1 L_1 K_2 L_2); \\
z_2  = \vu_3  \oplus \vu_4  \oplus \vu_5  \oplus \vu_8  \oplus \vu_9 \oplus \vu_{10} &{}& ( M_1 N_1 V_1 M_2 N_2 V_2); \\
z_3  = \vu_5  \oplus \vu_{10} &{}& ( V_1  V_2); \\
z_4  = \vu_1  \oplus \vu_2  \oplus \vu_5  \oplus \vu_6  \oplus \vu_7  \oplus \vu_{10} &{}& (K_1 L_1 V_1 K_2 L_2 V_2); \\
z_5  = \vu_3  \oplus \vu_4  \oplus \vu_8  \oplus \vu_9  &{}& (M_1 N_1 M_2 N_2); \\
z_6  = \vu_1  \oplus \vu_2  \oplus \vu_3  \oplus \vu_4  \oplus \vu_6  \oplus \vu_7  \oplus \vu_8  \oplus \vu_9 &{}& (K_1 L_1 M_1 N_1 K_2 L_2 M_2 N_2 ); \\
z_7  = \vu_1  \oplus \vu_3  \oplus \vu_6  \oplus \vu_8   &{}& (K_1 M_1 K_2 M_2 ); \\
z_8  = \vu_2  \oplus \vu_3  \oplus \vu_7  \oplus \vu_8   &{}& (L_1 M_1 L_2 M_2); \\
z_9  = \vu_1  \oplus \vu_4  \oplus \vu_6  \oplus \vu_9  &{}& (K_1 N_1 K_2 N_2); \\
z_{10}  = \vu_2  \oplus \vu_4  \oplus \vu_7  \oplus \vu_9  &{}& (L_1 N_1 L_2
N_2);\\
z_{11}  = z_3  \oplus z_6 &{}& (K_1 L_1
M_1 N_1 V_1 K_2 L_2 M_2 N_2 V_2).
\end{array}
\end{equation}
Выражения \eqref{eq1_4} добавляют еще 12 компонент
\begin{equation}\notag
\begin{array}{lll}
z_{12}  = \vu_1  \oplus \vu_3  \oplus \vu_5 \oplus \vu_7 \oplus \vu_9 &\hspace*{1mm} & (K_1 M_1 V_1 L_2 N_2) ;\\
z_{13}  = \vu_2 \oplus \vu_4  \oplus \vu_6  \oplus \vu_8  \oplus \vu_{10}
 &{}&
(L_1 N_1 K_2 M_2 V_2);\\
z_{14}  = \vu_1  \oplus \vu_4  \oplus \vu_7  \oplus \vu_8  \oplus \vu_{10}
 &{}&
(K_1 N_1 L_2 M_2 V_2);\\
z_{15}  = \vu_2  \oplus \vu_3  \oplus \vu_5  \oplus \vu_6  \oplus \vu_{9}
 &{}&
(L_1 M_1 V_1 K_2 N_2);\\
z_{16}  = \vu_1  \oplus \vu_4  \oplus \vu_5  \oplus \vu_7  \oplus \vu_{8}
 &{}&
(K_1 N_1 V_1 L_2 M_2) ;\\
z_{17}  = \vu_2  \oplus \vu_3  \oplus \vu_6  \oplus \vu_9  \oplus \vu_{10}
 &{}&
(L_1 M_1 K_2 N_2 V_2) ;\\
z_{18}  = \vu_1  \oplus \vu_3  \oplus \vu_7  \oplus \vu_9  \oplus \vu_{10}
 &{}&
(K_1 M_1 L_2 N_2 V_2) ;\\
z_{19}  = \vu_2  \oplus \vu_4  \oplus \vu_5  \oplus \vu_6  \oplus \vu_{8}
 &{}&
(L_1 N_1 V_1 K_2 M_2) ;\\
z_{20}  = \vu_1  \oplus \vu_2  \oplus \vu_5  \oplus \vu_8  \oplus \vu_{9}
 &{}&
(K_1 L_1 V_1 M_2 N_2) ;\\
z_{21}  = \vu_1  \oplus \vu_2  \oplus \vu_3  \oplus \vu_4  \oplus \vu_{10}
 &{}&
(K_1 L_1 M_1 N_1 V_2) ;\\
z_{22}  = \vu_5  \oplus \vu_6  \oplus \vu_7  \oplus \vu_8  \oplus \vu_{9}
 &{}&
(V_1 K_2 L_2 M_2 N_2) ;\\
z_{23}  = \vu_3  \oplus \vu_4  \oplus \vu_6  \oplus \vu_7  \oplus \vu_{10}
 &{}& (M_1 N_1 K_2 L_2 V_2).
\end{array}
\end{equation}

Зависимость фазовых переменных от переменных разделения описывается
линейной булевой вектор-функцией ($\bbZ$-линейным отображением)
$C:\bbZ^{10} \to \bbZ^{23}$. Выбирая в матрице $C$ строки с номерами
1, 10, 2, 3, 12 в качестве ведущих и переобозначая компоненты
$(z_1,z_{10},z_2,z_3,z_{12}) \to (z_1,z_2,z_3,z_4,z_5)$, приведем
эту матрицу к трапецевидной. Ненулевыми окажутся лишь выбранные
строки. Таким образом, $\rk C=5$ и редуцированная булева
вектор-функция имеет вид:
\begin{equation}\label{eq3_2}
\begin{array}{ll}
z_{1}  = (\vu_1  \oplus \vu_6)  \oplus (\vu_4 \oplus \vu_9),\\
z_{2}  = (\vu_2  \oplus \vu_7)  \oplus (\vu_4 \oplus \vu_9),\\
z_{3}  = (\vu_3  \oplus \vu_8)  \oplus (\vu_4 \oplus \vu_9),\\
z_{4}  = \vu_5  \oplus \vu_{10},\\
z_{5}  = \vu_6  \oplus \vu_7  \oplus \vu_8 \oplus \vu_9 \oplus \vu_{10}.
\end{array}
\end{equation}
Это выражение в определенном смысле окончательное, поскольку,
очевидно, количество компонент функции уменьшить для всех значений
параметров одновременно уже нельзя. Для дальнейшей работы со
столбцами необходимо их разбиение на группы, которое зависит от
номера области регулярности. Базируясь на информации, приведенной в
табл.~1, заполняем второй столбец табл.~2. Он содержит номера
аргументов второй группы, т.е. тех, которые относятся к
периодически меняющим знак радикалам. Здесь применен следующий
прием. В областях \ts{V}--\ts{X} переменная $t_2$ пересекает
бесконечность. В связи с этим все радикалы $R_{2 \gamma}$ следует
считать меняющими знак одновременно, что не отвечает существу дела.
Многочлен $S(\xi)$ имеет четную степень, т.е. $\infty$ не
является точкой ветвления для $\sqrt {S}$. Поступим следующим
образом. Поскольку всегда $|t_2| \geqslant v$, то $t_2 \ne 0$.
Положим
$$
\tilde R_{2 \gamma}= R_{2 \gamma}/ \sqrt{t_2}, \qquad
\vu_{5+\gamma}=\lsgn(\tilde R_{2 \gamma}) \qquad (\gamma=1, \ldots, 5).
$$
Выражения \eqref{eq1_3}, \eqref{eq1_4} устроены так, что после
перехода к новым радикалам особенность в точке $t_2=\infty$
исчезает, а при пересечении этого значения теперь никакие
подкоренные выражения в $\tilde R_{2\gamma}$ ни в нуль, ни в
бесконечность не обращаются, и аргументы $\vu_\gamma$ для переменной
$t_2$ относятся к соответствующей группе по конечным границам
промежутка изменения $t_2$.

\begin{table}[htp]
\centering
\small
\begin{tabular}{|c|c|l|c|c|}
\multicolumn{5}{r}{Таблица 2}\\
\hline
\begin{tabular}{c}Область\end{tabular}&\begin{tabular}{c}Вторая группа\end{tabular}&\begin{tabular}{c}Выражения компонент\end{tabular}
&\begin{tabular}{c}$(P,Q)$\end{tabular}&\begin{tabular}{c}$c(\cf)$\end{tabular}\\
\hline \ru  $\ts{I}^*$ & {$3\,5\,9\,{10}$} & $\quad \varnothing$ & (0,0)&{1}\\
\hline \ru $\ts{II}^*$ & {$3\,4\, {10}$}& $\quad \vu_5 \oplus \vu_8
\oplus \vu_9$& (1,0)&{2}\\
\hline \ru $\ts{III}\;$ & {$3\,5\,6\,9$} & $\quad \vu_1 \oplus \vu_4 \oplus \vu_7 \oplus \vu_8 \oplus \vu_{10}$ & (1,0)&{2}\\
\hline \ru $\ts{IV}\;$ & {$3\,4\,6\,{10}$} & $\quad \vu_1 \oplus \vu_2 \oplus \vu_5 \oplus \vu_8 \oplus \vu_{9}$ & (1,0)&{2}\\
\hline \ru $\ts{V}^*$ & {$5\,8\,9$} & $\quad \vu_3 \oplus \vu_4 \oplus \vu_6 \oplus \vu_7 \oplus \vu_{10}$&  {(1,0)}&{2}\\
\hline \ru $\ts{VI}^*$ & {$4\,5\,8\,{10}$} & $\quad \varnothing $&  {(0,0)}&{1}\\
\hline \ru $\ts{VII}^*$ & {$3\,4\,{10}$}& $ \quad \vu_5 \oplus \vu_6 \oplus \vu_7 \oplus \vu_8 \oplus \vu_{9}$&  {(1,0)}&{2}\\
\hline \ru $\ts{VIII}\;$ & {$2\,5\,8\,{9}$}& $\quad \vu_3 \oplus \vu_4 \oplus \vu_6 \oplus \vu_7 $&  {(1,0)}&{2}\\
\hline \ru $\ts{IX}\;$ & {$2\,4\,8\,{10}$}& $\quad \vu_1 \oplus \vu_3 \oplus \vu_5 \oplus \vu_7 \oplus \vu_{9}$&  {(1,0)}&{2}\\
\hline \ru $\ts{X}$ & {$2\,4\,{10}$} & $\left\{ \begin{array}{l}
\vu_1\oplus \vu_3 \oplus \vu_6 \oplus \vu_8\\
\vu_5 \oplus \vu_6 \oplus \vu_7 \oplus \vu_8 \oplus \vu_{9}\end{array}\right.$&  {(2,0)}&{4}\\
\hline \ru $\ts{XI}$ & {$3\,4\,7\,{10}$} & $\quad \vu_1 \oplus \vu_2 \oplus \vu_5 \oplus \vu_8 \oplus \vu_{9}$& (1,0)&{2}\\
\hline \ru $\ts{XII}$ & {$3\,5\,7\,{9}$} & $\quad \vu_2 \oplus \vu_4 \oplus \vu_6 \oplus \vu_8 \oplus \vu_{10}$& (1,0)&{2}\\
\hline \ru $\ts{XIII}$ & {$5\,7\,9$} & $\left\{ \begin{array}{l}\vu_1
\oplus \vu_3 \oplus \vu_6
\oplus \vu_8\\
\vu_2 \oplus \vu_4 \oplus \vu_6 \oplus \vu_8 \oplus \vu_{10}\end{array}\right.$& (2,0)&{4}\\
\hline
\end{tabular}
\end{table}

В табл. 2 звездочкой помечены области, в которых $\sigma<0$ и величина $\vk$
является чисто мнимой. Тогда пары $(K_1, L_1)$ и $(K_2,L_2)$ нужно
считать комплексно сопряженными, и произвольный знак можно выбирать
лишь у одного элемента каждой пары несмотря на то, что аргументы
$\vu_1,\vu_2, \vu_5,\vu_6$ все относятся к первой группе. Таким образом,
необходимо считать, что в этих случаях введены дополнительные
тождества
\begin{equation}\label{eq3_3}
\vu_1\oplus \vu_2 \equiv 0, \qquad \vu_6\oplus \vu_7 \equiv 0.
\end{equation}
Сопоставляя выражения компонент \eqref{eq3_2} с номерами аргументов
второй группы, видим, что в областях со звездочкой всегда
присутствует один из аргументов $\vu_4, \vu_9$, что обеспечивает
возможность одновременной замены значений $\vu_1,\vu_2$ в $z_1,z_2$, а
пара $\vu_6,\vu_7$ вообще входит лишь в виде суммы (в $z_5$) и на
результат ввиду \eqref{eq3_3} не влияет. На самом деле, структура
булевой вектор-функции автоматически исключает возможность получения
избыточных значений в случае наличия комплексно сопряженных
радикалов, поэтому результат не зависит от того, как произведена
редукция матрицы $C$ по строкам. Теперь в соответствии со списком
аргументов второй группы выполняем эквивалентные преобразования с
основной матрицей $C$ функции в виде \eqref{eq3_2}. А именно, если
аргумент $\vu_\gamma$ относится ко второй группе, то с помощью строки,
где он присутствует, обнуляем все остальные элементы этого столбца.
После этого такая ведущая строка в каноническом представлении
\eqref{eq2_4} окажется за пределами первых $P-Q$ строк и ее можно
отбросить как не влияющую на количество классов эквивалентности.
Остающиеся компоненты показаны в третьем столбце табл.~2.
Особенность матрицы такова, что во всех случаях \ts{I}--\ts{XIII}
все аргументы второй группы позволяют исключить ровно одну
компоненту. Таким образом, ранг оставшейся матрицы равен <<пять
минус количество аргументов второй группы>>, причем это всегда есть
первое число в паре рангов $(P,Q)$, так как аргументов второй группы
не остается. Соответственно, всегда в такой паре $Q=0$
(предпоследний столбец в табл.~2). В силу этого дальнейшие
преобразования со столбцами по приведению к каноническому виду
\eqref{eq2_4} можно уже не проводить. В случаях со звездочкой одна
из оставшихся компонент всегда имеет вид $\vu_1 \oplus \vu_2 \oplus \vu_6
\oplus \vu_7$, т.е. является константой в силу \eqref{eq3_3},
поэтому ее следует исключить, что понижает ранг еще на единицу (в
связи с этим в таблице возникает пустое множество компонент).
Окончательно, число $c(\cf)$ равно $2^P$ (см. последний столбец в
табл.~2). Напомним, что через $c(\cf)$ обозначено количество классов
эквивалентности для заданного значения $\cf=(s,\tau)$ постоянных
первых интегралов -- количество связных компонент (торов
Лиувилля) интегрального многообразия $\FF$. Результат показан на
рис.~\ref{fig_harl2},\,{\it a}. Исследование регулярной фазовой
топологии задачи завершено.

Рассмотрим критические случаи -- отрезки бифуркационного множества
между областями. Поскольку наша цель -- продемонстрировать технику,
то для краткости случаи, известные из других источников, просто
перечислим. Случаи $\chi=0$ отвечают в точности особым периодическим
решениям случая Богоявленского \cite{Kh36}, причем константа
интеграла $s$ -- это использованный в указанной работе параметр, а
величина $\tau$ выражается через постоянную частного интеграла
Богоявленского. Поэтому в каждую точку множества $\chi=0$ переходят
две точки кривых $\delta_1, \delta_2, \delta_3$ работы \cite{Kh36}.
Таким образом, на криволинейных участках разделяющего множества
имеем следующие критические поверхности: $2 S^1$ при $s\in (-b,0)
\cup (0,b)$, $4 S^1$ при $s\in (a,+\infty)$. Полупрямая $\tau=0,
s>0$ отвечает пересечению с критической подсистемой, изученной в
\cite{KhSav}. В ней этому множеству соответствует случай $\ell=0,
m<0$ и имеется такая связь использованных параметров: $ m=-
1/2s$. Следовательно интегральные поверхности
для $\tau=0$ таковы: $2\bbT^2$ при $s\in (0,b) \cup (b,a)$,
$4\bbT^2$ при $s\in (a,+\infty)$. Этот случай является
полурегулярным, так как реальных бифуркаций не происходит, возникает
лишь особенность в выражениях \eqref{eq1_3}, \eqref{eq1_4},
которую можно устранить заменой переменных.

Для примера рассмотрим более подробно переходы, возникающие при
круговом обходе критической точки типа <<седло-седло>>.

\begin{figure}[htp]
\centering
\includegraphics[width=125mm,keepaspectratio]{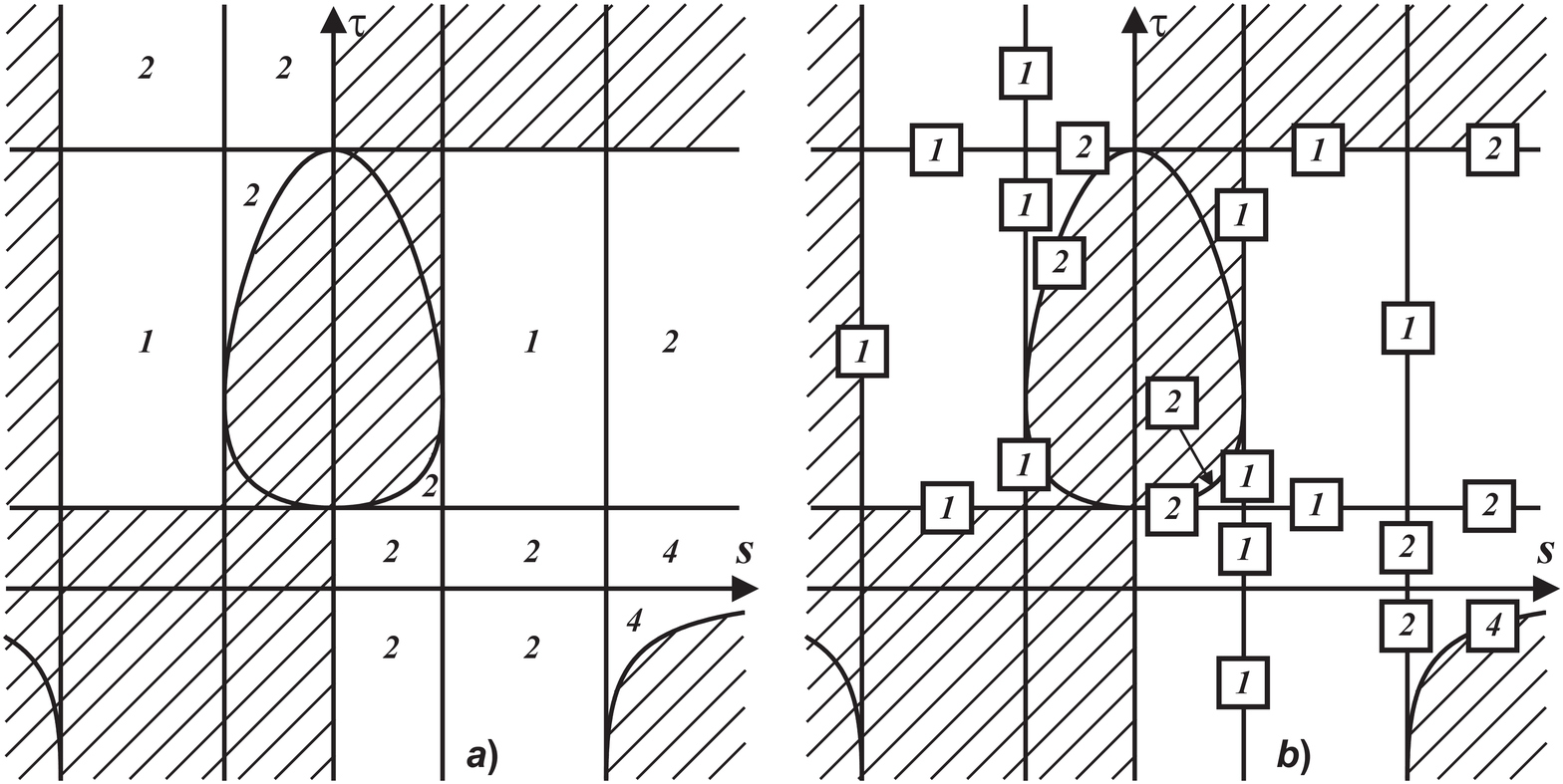}
\caption{Количество торов и связность критических
поверхностей.}\label{fig_harl2}
\end{figure}

При переходе $\ts{I}{\to}\ts{II}$ (см. рис. 1) возникает кратный корень
$-\vo=\no$. Во вторую группу попадают радикалы $3, 4, 5, 9, 10$.
Переменная $\vu_5$ позволяет исключить единственную нетривиальную
компоненту области $\ts{II}^*$ (см. табл.~2). Получаем пару
$(P,Q)=(0,0)$, поэтому $c(\cf)=1$. Но перестройка здесь $\bbT^2 \to
2\bbT^2$. Следовательно, $\FF=(S^1 \vee S^1)\times S^1$.

При переходе $\ts{I}{\to}\ts{III}$ возникает кратный корень
$\varkappa=0$. Во вторую группу попадают радикалы $3,5,6,7,9$.
Переменная $\vu_7$ позволяет исключить единственную нетривиальную
компоненту области $\ts{III}$. Получаем $(P,Q)=(0,0)$, $c(\cf)=1$.
Перестройка $\bbT^2 \to 2\bbT^2$, значит, $\FF=(S^1 \vee S^1)\times
S^1$.

При переходе $\ts{III}{\to}\ts{IV}$ возникает кратный корень
$-\vo=\no$, лежащий строго между $\mo$ и $-\varkappa$. Во вторую
группу попадают радикалы $3,4,5,6,8,10$. По отношению к области
\ts{III} добавились аргументы $\vu_4,\vu_8$, любой из них позволяет
исключить единственную компоненту функции. По отношению к области
\ts{IV} добавились аргументы $\vu_5,\vu_8$, и снова любой из них
позволяет исключить единственную компоненту функции. В результате
$(P,Q)=(0,0)$, $c(\cf)=1$. Перестройка $2\bbT^2 \to 2\bbT^2$,
значит, $\FF=(S^1 \ddot \vee S^1)\times S^1$,  где $S^1 \ddot \vee S^1$ -- пара окружностей, пересекающихся по двум
точкам.

При переходе $\ts{II}{\to}\ts{IV}$ возникает кратный корень
$\varkappa=0$ между $-\vo$ и $\vo$. Во вторую группу попадают
радикалы $3,4,6,7,10$. По отношению к области \ts{II} добавились
аргументы $\vu_6,\vu_7$, которые входили в компоненту-константу, поэтому
это не повлияло на результат. По отношению к области \ts{IV}
добавился аргумент $\vu_7$, и он не входит в оставшуюся согласно
табл.~2 компоненту функции. Поэтому результат такой же, как и в
самих областях -- $(P,Q)=(1,0)$, $c(\cf)=2$. Однако, здесь
перестройка по количеству торов такая же, как в предыдущем случае
$2\bbT^2 \to 2\bbT^2$. Поэтому критическая поверхность есть
$\FF=2(S^1 \vee S^1)*S^1$ (каждая компонента связности -- косое
произведение <<восьмерки>> на окружность).

Аналогично рассматриваются и все остальные случаи. На
рис.~\ref{fig_harl2},\,{\it b}  показано количество связных компонент
критических поверхностей для всех участков бифуркационной диаграммы,
из чего, благодаря найденному ранее количеству торов Лиувилля в
регулярных областях, однозначно восстанавливается и топологический
тип любой критической поверхности.

\end{document}